\documentclass[12pt,a4paper]{article}
\usepackage{amssymb}
\usepackage{amsmath}
\usepackage{amsfonts}

\newtheorem{theorem}{Theorem}

\newtheorem{definition}[theorem]{Definition}
\newtheorem{example}[theorem]{Example}

\newtheorem{proposition}[theorem]{Proposition}
\newtheorem{remark}[theorem]{Remark}

\newenvironment{proof}[1][Proof]{\noindent\textbf{#1.} }{\ \rule{0.5em}{0.5em}}
\input{tcilatex}
\begin{document}

\title{Spectral parameter power series method for discontinuous coefficients}
\author{Herminio Blancarte$^{\text{1}}$, Hugo M. Campos$^{\text{1,2}}$, Kira
V. Khmelnytskaya$^{\text{1}}$ \\
$^{\text{1}}${\small Faculty of Engineering, Autonomous University of
Queretaro, }\\
{\small Centro Universitario}, {\small Cerro de las Campanas s/n, C.P.
76010, Santiago de Quer\'{e}taro, Qro., M\'{e}xico}\\
$^{\text{2}}${\small Department of Mathematics, FCFM, Benem\'{e}rita
Universidad Aut\'{o}noma de Puebla, }\\
{\small Av. San Claudio y 18 sur San Manuel CU, CP. 72570 Puebla, Pue.,
Mexico}\\
{\small email: herbs@uaq.mx, hugomcampos@hotmail.com, khmel@uaq.edu.mx%
\thanks{%
Research was supported by CONACYT, Mexico via the research projects 166141
and 176987. The second named author additionally acknowledges the support by
FCT, Portugal.}}}
\maketitle

\begin{abstract}
Let $(a,b)$ be a finite interval and $1/p$, $q$, $r\in L^{1}[a,b]$. We show
that a general solution (in the weak sense) of the equation $(pu^{\prime
})^{\prime }+qu=\lambda ru$ on $(a,b)$ can be constructed in terms of power
series of the spectral parameter $\lambda $. The series converge uniformly
on $[a,b]$ and the corresponding coefficients are constructed by means of a
simple recursive procedure. We use this representation to solve different
types of eigenvalue problems. Several numerical tests are discussed.
\end{abstract}

\section{Introduction}

The spectral parameter power series (SPPS) method proposed in \cite{Krv} as
an application of pseudoanalytic function theory and developed into a
numerical technique in \cite{[KP]} rapidly became an important and efficient
tool for solving a variety of problems involving Sturm-Liouville equations.
In \cite{BarBlKr}, \cite{BarKrR}, \cite{CastKrOR}, \cite{KhKrR}, \cite{KrTV}%
, \cite{KrV} the SPPS method was used for solving spectral problems for
Sturm-Liouville equations, in \cite{CKKO} the electromagnetic scattering
problem was studied, in \cite{KKB2012} for fourth-order Sturm-Liouville
equations and in \cite{CKT} to study interesting mapping properties of
transmutation operators.

In previous publications dedicated to the SPPS method the coefficients of
the differential equations were assumed to be continuous functions. Only in 
\cite{KhS} with a short explanation there were considered models involving
discontinuous coefficients. However, obviously, many different applications
require considering piecewise continuos or only integrable coefficients.

In the present paper we obtain the SPPS representation for solutions of the
Sturm-Liouville equation with coefficients from the space $L_{1}\left[ a,b%
\right] $ (the precise conditions can be found in Section 3). The proof in
this case required to find different ways to obtain estimates for the
summands of the SPPS series as well as to prove the convergence of the
series in the corresponding norms (especially for the series defining the
derivatives of the solutions). Though the proofs resulted to be quite
different and more elaborate as compared to the SPPS representations in the
case of continuous coefficients, in general the SPPS approach did not suffer
any serious modification and the main result of this work consists in the
fact that the SPPS method is now available in the much more general
situation of discontinuous coefficients.

The paper is structured as follows. In Section 2 we introduce preliminary
facts and definitions concerning weak solutions of the Sturm-Liouville
equation and absolutely continuous functions and we also establish some
auxiliary results. In Section 3 we prove the main result of the paper, the
SPPS representation for solutions of the Sturm-Liouville equation with
discontinuous coefficients (Theorem \ref{main_theorem}). Section 4 is
dedicated to the numerical implementation of the SPPS method and numerical
examples involving spectral problems for Sturm-Liouville equations with
discontinuous (and, in general, complex) coefficients. Excellent performance
of the method is illustrated.

\section{Preliminaries}

First let us introduce some classical function spaces that will be needed
throughout the paper. Let $L^{1}[a,b]$ be the Lebesgue space of absolutely
integrable functions on $[a,b]$. As usual\ the Sobolev space $W^{1,1}[a,b]$
is formed by functions $u\in L^{1}[a,b]$ for which there exists $g\in
L^{1}[a,b]$ such that 
\begin{equation*}
\int\limits_{a}^{b}u\varphi ^{\prime }dt=-\int\limits_{a}^{b}g\varphi
dt,\qquad \forall \varphi \in C_{c}^{\infty }(a,b),
\end{equation*}%
where $C_{c}^{\infty }(a,b)$ is the space of infinitely differentiable
functions on $(a,b)$ with compact support in $(a,b)$. The function $g$ is
called distributional (or weak) derivative of $u$ and is denoted by $%
u^{\prime }$.

Denote the set of complex-valued absolutely continuous functions on $[a,b]$
by $AC\left[ a,b\right] $. We recall that a function $V$ belongs to $AC\left[
a,b\right] $ iff there exists $v\in L^{1}[a,b]$ and $x_{0}\in \lbrack a,b]$
such that%
\begin{equation}
V(x)=\int_{x_{0}}^{x}v(t)dt+V(x_{0})\text{, \ \ }\forall \,x\in \lbrack a,b]
\label{AC}
\end{equation}%
(see, e.g., \cite{Benedetto}, \cite{Natanson}, \cite{Kolmogorov}). Thus, if $%
V\in AC[a,b]$ then $V\in C[a,b]$ and $V$ is a.e. differentiable with $%
V^{\prime }(x)=v(x)$ on $[a,b]$. Moreover, the usual derivative of a
function from $AC[a,b]$ coincides a.e. with its distributional derivative 
\cite{Kolmogorov}, thus $AC[a,b]\subseteq W^{1,1}[a,b]$. In fact the linear
spaces $W^{1,1}[a,b]$ and $AC[a,b]$ coincide (see \cite{Haim Brezis},
Theorem 8.2) in the following sense: $u\in W^{1,1}[a,b]$ iff there exists $%
V\in AC[a,b]$ such that $u=V$ a.e. on $[a,b]$.

Let us introduce the concept of weak solution corresponding to the equation 
\begin{equation}
(pu^{\prime })^{\prime }+qu=0.  \label{SOE}
\end{equation}

First assume that $p\in C^{1}[a,b]$ and $q\in C[a,b]$. Let $u\in C^{2}[a,b]$
be a classical solution of (\ref{SOE}). Multiplying (\ref{SOE}) by $\varphi
\in C_{c}^{\infty }[a,b]$ and integrating by parts we arrive at the equality%
\begin{equation}
-\int\limits_{a}^{b}pu^{\prime }\varphi ^{\prime
}+\int\limits_{a}^{b}qu\varphi =0,\quad \forall \varphi \in C_{c}^{\infty
}[a,b]  \label{weakSOE}
\end{equation}%
where we have used that $\varphi (a)=\varphi (b)=0.$

\begin{definition}
\label{Def weak}Let $p$ be a measurable function on $[a,b]$ and $q\in
L^{1}[a,b]$. A function $u\in AC[a,b]=W^{1,1}[a,b]$ is called a weak
solution of (\ref{SOE}) if $pu^{\prime }\in L^{1}[a,b]$ and $u$ satisfies (%
\ref{weakSOE}).
\end{definition}

\begin{proposition}
\label{Prop_weak}Under the conditions of the above definition a function $%
u\in AC[a,b]$ is a weak solution of (\ref{SOE}) iff $pu^{\prime }\in AC[a,b]$
and (\ref{SOE}) is satisfied a.e. on $[a,b].$
\end{proposition}

\begin{proof}
Let $u\in AC[a,b]$ be a weak solution of (\ref{SOE}). Equality (\ref{weakSOE}%
) means that the distributional derivative of $pu^{\prime }\in L^{1}[a,b]$
is $-qu\in L^{1}[a,b]$. Thus $pu^{\prime }\in AC[a,b]$ and (\ref{SOE}) is
satisfied a.e. on $[a,b].$

The opposite statement follows from the fact that the usual derivative of a
function from $AC[a,b]$ coincides with its distributional derivative.
\end{proof}

The following result will be useful further.

\begin{proposition}
\label{Prop_conv}Let $\left\{ V_{n}\right\} _{n=0}^{\infty }$ be a sequence
of absolutely continuous functions on $[a,b]$. If the series $%
\sum_{n=0}^{\infty }V_{n}(x_{0})$ converges at some $x_{0}\in \left[ a,b%
\right] $ and the series $\sum_{n=0}^{\infty }V_{n}^{\prime }$ converges to $%
v\in L^{1}[a,b]$ in the norm of $L^{1}[a,b]$, then $\sum_{n=0}^{\infty
}V_{n} $ converges uniformly to $V\in AC\left[ a,b\right] $ and $V^{\prime
}=v$ a.e.
\end{proposition}

This is in fact a well-known result when the functions $V_{n}$ are
continuously differentiable \cite{Apostol}. Also, it is not hard to see that
Proposition \ref{Prop_conv} is equivalent to the fact that the space $%
AC[a,b] $ equipped with the norm given by $\left\Vert u\right\Vert
=\left\vert u(x_{0})\right\vert +\left\Vert u^{\prime }\right\Vert _{L^{1}}$
is a Banach space. This result is mentioned and used in many sources (see,
e.g., \cite{Matkowski}), but since we know no reference where a detailed
rigorous proof can be found, we provide one here.

\begin{proof}
Consider the absolutely continuous function 
\begin{equation}
V(x):=\int_{x_{0}}^{x}v(t)dt+\sum_{n=0}^{\infty }V_{n}(x_{0}).  \label{defv}
\end{equation}%
Using the representation (\ref{AC}) for $V_{n}$ we find that%
\begin{equation*}
\left\vert V(x)-\sum_{n=0}^{N}V_{n}(x)\right\vert =\left\vert
\int_{x_{0}}^{x}v(t)dt+V(x_{0})-\sum_{n=0}^{N}\left(
\int_{x_{0}}^{x}V_{n}^{\prime }(t)dt+V_{n}(x_{0})\right) \right\vert \leq
\end{equation*}%
\begin{equation*}
\leqslant \left\vert V(x_{0})-\sum_{n=0}^{N}V_{n}(x_{0})\right\vert
+\left\vert
\int_{x_{0}}^{x}v(t)dt-\sum_{n=0}^{N}\int_{x_{0}}^{x}V_{n}^{\prime
}(t)dt\right\vert =
\end{equation*}%
\begin{equation*}
=\left\vert V(x_{0})-\sum_{n=0}^{N}V_{n}(x_{0})\right\vert +\left\vert
\int_{a}^{b}\left( v(t)-\sum_{n=0}^{N}V_{n}^{\prime }(t)\right)
dt\right\vert \leq
\end{equation*}%
\begin{equation*}
\leq \left\vert V(x_{0})-\sum_{n=0}^{N}V_{n}(x_{0})\right\vert
+\int_{a}^{b}\left\vert v(t)-\sum_{n=0}^{N}V_{n}^{\prime }(t)\right\vert dt
\end{equation*}%
When $N\rightarrow \infty $ ,\ the last expression tends to zero since $%
\sum_{k=0}^{\infty }V_{n}(x_{0})=V(x_{0})$ and $\sum_{n=0}^{\infty
}V_{n}^{\prime }$ converges to $v$\ in the norm of $L^{1}[a,b]$. Thus, $%
\sum_{n=0}^{\infty }V_{n}(x)$ converges uniformly to $V(x)$ and from (\ref%
{defv}) it follows that $V^{\prime }(x)=v(x)=\sum_{k=0}^{\infty
}V_{n}^{\prime }(x)$ a.e..
\end{proof}

\section{SPPS representation for solutions of the Sturm-Liouville equation}

Consider the Sturm-Liouville equation%
\begin{equation}
(pu^{\prime })^{\prime }+qu=\lambda ru  \label{SLlambda}
\end{equation}%
on some finite interval $(a,b)$ where $1/p$, $q$, $r\in L^{1}[a,b]$ and $%
\lambda \in \mathbb{C}$ is the spectral parameter. Following \cite{Zettl} we
define the operator $L[u]:=(pu^{\prime })^{\prime }+qu$ with domain of
definition given by $D_{L}=\left\{ u:u\in AC[a,b]\text{ and }pu^{\prime }\in
AC[a,b]\right\} $. By a solution of equation (\ref{SLlambda}) we mean a
function $u$ which belongs to $D_{L}$ and satisfies $L[u]=\lambda ru$ $a.e.$
on $(a,b)$. It follows from Proposition \ref{Prop_weak} that this definition
of solution is equivalent to that one introduced in Definition \ref{Def weak}%
.

\begin{proposition}[Polya Factorization]
\label{Polya_fact} Let $f$ be a nonvanishing solution of $L[f]=0$. Then for $%
u\in D_{L}$ the following equality%
\begin{equation}
L[u]=\frac{1}{f}\frac{d}{dx}pf^{2}\frac{d}{dx}\frac{1}{f}u\text{ }
\label{Polyafact}
\end{equation}%
holds $a.e.$ on $(a,b)$.
\end{proposition}

Let $f\in AC[a,b]$ and $f(x)\neq 0$ for any $x\in \lbrack a,b]$. Following 
\cite{[KP]} let us introduce two families of functions $\widetilde{X}%
^{(n)}(x)$ and $X^{(n)}(x)$ by the recursive equalities%
\begin{equation*}
\widetilde{X}^{(0)}\equiv 1,\qquad X^{(0)}\equiv 1,
\end{equation*}

\begin{equation}
\widetilde{X}^{(n)}(x)=%
\begin{cases}
\int_{x_{0}}^{x}\widetilde{X}^{(n-1)}(s)r(s)f^{2}(s)ds\text{,}\qquad \text{%
for\thinspace an\thinspace odd}\,\ n \\ 
\\ 
\int_{x_{0}}^{x}\widetilde{X}^{(n-1)}(s)\frac{ds}{p(s)f^{2}(s)}\text{,}%
\qquad \text{for\thinspace an\thinspace even }\,n%
\end{cases}
\label{K1}
\end{equation}%
\begin{equation}
X^{(n)}(x)=%
\begin{cases}
\int_{x_{0}}^{x}X^{(n-1)}(s)\frac{ds}{p(s)f^{2}(s)}\text{,}\qquad \text{%
for\thinspace an\thinspace odd }n \\ 
\\ 
\int_{x_{0}}^{x}X^{(n-1)}(s)r(s)f^{2}(s)ds\text{,}\qquad \text{for\thinspace
an\thinspace even}\,n~.%
\end{cases}
\label{K2}
\end{equation}%
where $x_{0}$ is an arbitrary point in $[a,b]$. Below we show that the
introduced families of functions are closely related to the Sturm-Liouville
equation (\ref{SLlambda}). For this the following proposition will be used.

\begin{proposition}
\label{Xestimates}Under the above conditions the functions $\widetilde{X}%
^{(n)}$ and $X^{(n)}$ belong to $AC[a,b]$ and the following estimates hold%
\begin{equation}
\left\vert X^{(2n)}(x)\right\vert ,\left\vert \widetilde{X}%
^{(2n)}(x)\right\vert \leq \frac{(C_{1}C_{2})^{n}}{n!n!},  \label{estim_even}
\end{equation}%
\begin{equation}
\left\vert X^{(2n-1)}(x)\right\vert \leq \frac{C_{1}^{n}}{n!}\frac{%
C_{2}^{n-1}}{\left( n-1\right) !},\quad \left\vert \widetilde{X}%
^{(2n-1)}(x)\right\vert \leq \frac{C_{1}^{n-1}}{\left( n-1\right) !}\frac{%
C_{2}^{n}}{n!},  \label{estim_odd}
\end{equation}%
where $C_{1}=\left\Vert 1/(pf^{2})\right\Vert _{L^{1}[a,b]}$, $%
C_{2}=\left\Vert rf^{2}\right\Vert _{L^{1}[a,b]}$ and $n=0,1,2,...$
\end{proposition}

\begin{proof}
The absolute continuity of the functions $\widetilde{X}^{(n)}(x)$ and $%
X^{(n)}(x)$\ follows immediately from the summability of the integrands \cite%
{Kolmogorov}. Let us prove the inequalities when $x\geq x_{0}$, for $x<x_{0}$
a similar reasoning can be applied. First let us obtain the following
auxiliary inequalities 
\begin{equation}
\left\vert X^{(2n)}(x)\right\vert ,\left\vert \widetilde{X}%
^{(2n)}(x)\right\vert \leq \frac{P(x)^{n}R(x)^{n}}{n!n!},
\label{estim_even_pr}
\end{equation}

\begin{equation}
\left\vert X^{(2n-1)}(x)\right\vert \leq \frac{P(x)^{n}}{n!}\frac{R(x)^{n-1}%
}{\left( n-1\right) !},\quad \left\vert \widetilde{X}^{(2n-1)}(x)\right\vert
\leq \frac{P(x)^{n-1}}{\left( n-1\right) !}\frac{R(x)^{n}}{n!},
\label{estim _odd_pr}
\end{equation}%
$x\geq x_{0}$, where%
\begin{equation}
P(x)=\int_{x_{0}}^{x}\frac{ds}{\left\vert p(s)f^{2}(s)\right\vert },\qquad
R(x)=\int_{x_{0}}^{x}\left\vert r(s)f^{2}(s)\right\vert ds.
\label{defPand R}
\end{equation}%
As $P$ and $R$ are monotonically increasing then $P(x)\leq C_{1}$ and $%
R(x)\leq C_{2}$ thus (\ref{estim_even}) and (\ref{estim_odd}) follows from (%
\ref{estim_even_pr}) and (\ref{estim _odd_pr}). The proof of (\ref%
{estim_even_pr}) and (\ref{estim _odd_pr}) is by induction. First we
consider $\widetilde{X}^{(2n)}$. For $n=0$ the estimate (\ref{estim_even_pr}%
) is trivial. Suppose that (\ref{estim_even_pr}) is valid for $n=k-1$,%
\begin{equation*}
\left\vert \widetilde{X}^{(2k-2)}(x)\right\vert \leq \frac{%
P(x)^{k-1}R(x)^{k-1}}{\left( k-1\right) !\left( k-1\right) !}.
\end{equation*}%
Then for $n=k$ we have

\begin{equation*}
\left\vert \widetilde{X}^{(2k)}(x)\right\vert \leq \int_{x_{0}}^{x}\frac{%
\left\vert \widetilde{X}^{(2k-1)}(s)\right\vert }{\left\vert
p(s)f^{2}(s)\right\vert }ds\leq \int_{x_{0}}^{x}\frac{1}{\left\vert
p(s)f^{2}(s)\right\vert }\int_{x_{0}}^{s}\left\vert \widetilde{X}%
^{(2k-2)}(t)\right\vert \left\vert r(t)f^{2}(t)\right\vert dtds\leq
\end{equation*}%
\begin{eqnarray*}
\int_{x_{0}}^{x}\frac{1}{\left\vert p(s)f^{2}(s)\right\vert }\int_{x_{0}}^{s}%
\frac{P(t)^{k-1}R(t)^{k-1}}{\left( k-1\right) !\left( k-1\right) !}%
\left\vert r(t)f^{2}(t)\right\vert dtds &\leq & \\
\frac{1}{\left( k-1\right) !\left( k-1\right) !}\int_{x_{0}}^{x}\frac{%
P(s)^{k-1}}{\left\vert p(s)f^{2}(s)\right\vert }\int_{x_{0}}^{s}R(t)^{k-1}%
\left\vert r(t)f^{2}(t)\right\vert dtds &\leq &
\end{eqnarray*}%
\begin{eqnarray*}
\frac{1}{\left( k-1\right) !k!}\int_{x_{0}}^{x}\frac{P(s)^{k-1}}{\left\vert
p(s)f^{2}(s)\right\vert }R(s)^{k}ds &\leq & \\
\frac{R(x)^{k}}{\left( k-1\right) !k!}\int_{x_{0}}^{x}\frac{P(s)^{k-1}}{%
\left\vert p(s)f^{2}(s)\right\vert }ds &=&\frac{R(x)^{k}P(x)^{k}}{k!k!}
\end{eqnarray*}%
where we have used the equalities $dR(t)=\left\vert r(t)f^{2}(t)\right\vert
dt$ and $dP(t)=\frac{1}{\left\vert p(t)f^{2}(t)\right\vert }dt$ following
from (\ref{defPand R}).

The proof for $\widetilde{X}^{(2n-1)}(x)$ is similar. For $n=1$ the
inequality (\ref{estim _odd_pr}) holds 
\begin{equation*}
\left\vert \widetilde{X}^{(1)}(x)\right\vert \leq \int_{x_{0}}^{x}\left\vert
r(s)f^{2}(s)\right\vert ds=R(x),
\end{equation*}%
Let for $n=k-1$ (\ref{estim _odd_pr}) \ be true 
\begin{equation*}
\left\vert \widetilde{X}^{(2k-3)}(x)\right\vert \leq \frac{P(x)^{k-2}}{%
\left( k-2\right) !}\frac{R(x)^{k-1}}{\left( k-1\right) !}.
\end{equation*}%
Then for $n=k$ we have%
\begin{equation*}
\left\vert \widetilde{X}^{(2k-1)}(x)\right\vert \leq
\int_{x_{0}}^{x}\left\vert r(s)f^{2}(s)\right\vert \int_{x_{0}}^{s}\frac{%
\left\vert \widetilde{X}^{(2k-3)}(t)\right\vert }{\left\vert
p(t)f^{2}(t)\right\vert }dtds\leq
\end{equation*}%
\begin{equation*}
\frac{1}{\left( k-2\right) !\left( k-1\right) !}\int_{x_{0}}^{x}\left\vert
r(s)f^{2}(s)\right\vert \int_{x_{0}}^{s}\frac{P(t)^{k-2}R(t)^{k-1}}{%
\left\vert p(t)f^{2}(t)\right\vert }dtds\leq
\end{equation*}%
\begin{equation*}
\frac{1}{\left( k-1\right) !\left( k-1\right) !}\int_{x_{0}}^{x}\left\vert
r(s)f^{2}(s)\right\vert R(s)^{k-1}P(s)^{k-1}ds\leq \frac{P(x)^{k-1}}{\left(
k-1\right) !}\frac{R(x)^{k}}{k!}
\end{equation*}

The estimates for $X^{(n)}$ can be shown similarly.
\end{proof}

\begin{remark}
\label{Remark Xk}When additionally $1/p$ and $r\in L^{\infty }[a,b]$,
stronger estimates can be obtained for $\widetilde{X}^{(n)}$ and $X^{(n)}$
(as those obtained in \cite{[KP]} for continuous coefficients). For example,%
\begin{equation*}
\left\vert X^{(2n)}(x)\right\vert \leq \left\Vert rf^{2}\right\Vert
_{L^{\infty }[a,b]}^{n}\left\Vert 1/(pf^{2})\right\Vert _{L^{\infty
}[a,b]}^{n}\frac{(b-a)^{2n}}{\left( 2n\right) !}.
\end{equation*}
\end{remark}

The following theorem generalizes the result from \cite{[KP]} onto the case
when $1/p$, $q$ and $r$ are integrable on $[a,b]$.

\begin{theorem}
\label{main_theorem}Let $p$, $q$ and $r$ be complex-valued functions of a
real variable $x\in \lbrack a,b]$ such that $1/p$, $q$ and $r\in L^{1}[a,b]$%
. Let $f$ be a nonvanishing weak solution of the equation%
\begin{equation}
(pf^{\prime })^{\prime }+qf=0  \label{SLhom}
\end{equation}%
on $[a,b]$. Then the general weak solution $u$ of the equation (\ref%
{SLlambda}) on $[a,b]$ has the form 
\begin{equation*}
u=c_{1}u_{1}+c_{2}u_{2},
\end{equation*}%
where $c_{1}$ and $c_{2}$ are arbitrary complex constants,%
\begin{equation}
u_{1}=f\sum_{k=0}^{\infty }\lambda ^{k}\;\widetilde{X}^{(2k)}\text{ \ and }%
u_{2}=f\sum_{k=0}^{\infty }\lambda ^{k}\;X^{(2k+1)}  \label{SPPSsolutions}
\end{equation}%
and both series converge uniformly on $[a,b]$.
\end{theorem}

\begin{proof}
First let us prove that both series $\sum_{k=0}^{\infty }\lambda ^{k}\;%
\widetilde{X}^{(2k)}$ and $\sum_{k=1}^{\infty }\lambda ^{k}\;\widetilde{X}%
^{(2k-1)}$ converge uniformly to the functions $v\ $and $w\in AC[a,b]$
respectively, and that the series obtained by term-wise differentiation
converge to the functions $v^{\prime }$ and $w^{\prime }$ respectively in
the space $L^{1}[a,b]$. Due to Proposition \ref{Xestimates} we have%
\begin{equation*}
\sum_{k=0}^{\infty }\left\vert \lambda ^{k}\;\widetilde{X}^{(2k)}\right\vert
\leq \sum_{k=0}^{\infty }\frac{C^{k}}{k!k!}<\infty ,
\end{equation*}%
\begin{equation*}
\sum_{k=1}^{\infty }\left\vert \lambda ^{k}\;\widetilde{X}%
^{(2k-1)}\right\vert \leq \left\vert \lambda \right\vert
C_{2}\sum_{k=1}^{\infty }\frac{C^{k-1}}{(k-1)!k!}<\infty ,
\end{equation*}%
where $C=\left\vert \lambda \right\vert C_{1}C_{2}$ , and the constants $%
C_{1}$, $C_{2}$ are defined in Proposition \ref{Xestimates}. Thus, the
uniform convergence of the series defining the functions $v$ and $w$ follows
from the Weierstrass $M$-test. In order to prove that the series obtained
from $\sum_{k=0}^{\infty }\lambda ^{k}\;\widetilde{X}^{(2k)}$ by term-wise
differentiation converges in the space $L^{1}[a,b]$ to $v^{\prime }=\dfrac{w%
}{pf^{2}}$ we consider the $L^{1}$-norm of the difference 
\begin{eqnarray*}
\left\Vert \sum_{k=0}^{N}\lambda ^{k}\left( \widetilde{X}^{(2k)}\right)
^{\prime }-\frac{w}{pf^{2}}\right\Vert _{L^{1}} &=&\int_{a}^{b}\left\vert
\sum_{k=1}^{N}\lambda ^{k}\;\frac{\widetilde{X}^{(2k-1)}}{pf^{2}}-\frac{w}{%
pf^{2}}\right\vert dx \\
&\leq &\sup_{[a,b]}\left\vert \sum_{k=1}^{N}\lambda ^{k}\;\widetilde{X}%
^{(2k-1)}-w\right\vert C_{1}\rightarrow 0
\end{eqnarray*}%
as $N\rightarrow \infty .$ From Proposition \ref{Prop_conv} we conclude that 
$v\in AC[a,b]$, the series defining $v$ can be term-wise differentiated and $%
v^{\prime }=\frac{w}{pf^{2}}$ a.e. The corresponding result for the function 
$w$ is proved similarly.

Next we prove that $u_{1}$ is a solution of (\ref{SLlambda}). The functions $%
u_{1}$ and%
\begin{equation*}
pu_{1}^{\prime }=pf^{\prime }\sum_{k=0}^{\infty }\lambda ^{k}\;\widetilde{X}%
^{(2k)}+\frac{1}{f}\sum_{k=1}^{\infty }\lambda ^{k}\;\widetilde{X}^{(2k-1)}
\end{equation*}%
are absolutely continuous on $[a,b]$. Moreover, as $f$ is a nonvanishing
solution of (\ref{SLhom}) the operator $L$ admits the Polya factorization (%
\ref{Polyafact}) and hence we have%
\begin{equation*}
L\left[ u_{1}\right] =\frac{1}{f}\frac{d}{dx}pf^{2}\frac{d}{dx}%
\sum_{k=0}^{\infty }\lambda ^{k}\;\widetilde{X}^{(2k)}=\frac{1}{f}\frac{d}{dx%
}\sum_{k=1}^{\infty }\lambda ^{k}\;\widetilde{X}^{(2k-1)}=rf\sum_{k=1}^{%
\infty }\lambda ^{k}\;\widetilde{X}^{(2k-2)}=\lambda ru_{1}.
\end{equation*}%
In a similar way we can check that $u_{2}$ satisfies (\ref{SLlambda}) as
well.

The last step is to verify that the generalized Wronskian of $u_{1}$ and $%
u_{2}$ 
\begin{equation*}
pW(u_{1},u_{2})=p(x)(u_{1}(x)u_{2}^{\prime }(x)-u_{2}(x)u_{1}^{\prime }(x))
\end{equation*}
is different from zero at some point (this is equivalent to show that the
solutions $u_{1}$ and $u_{2}$ are linearly independent \cite{Zettl}). As $%
\widetilde{X}^{(n)}(x_{0})=X^{(n)}(x_{0})=0$ for any $n\geq 1$ it is easy to
see that%
\begin{equation*}
\begin{array}{ll}
u_{1}(x_{0})=f(x_{0}),\medskip & p(x)u_{1}^{\prime
}(x)|_{x=x_{0}}=p(x)f^{\prime }(x)|_{x=x_{0}} \\ 
u_{2}(x_{0})=0, & p(x)u_{2}^{\prime }(x)|_{x=x_{0}}=\frac{1}{f(x_{0})}.%
\end{array}%
\end{equation*}%
Thus $pW(u_{1},u_{2})=1.$
\end{proof}

\bigskip

\begin{remark}
\label{Shift lambda}The SPPS representation for solutions (\ref%
{SPPSsolutions}) established in Theorem \ref{main_theorem} is based on a
particular solution $f$ corresponding to $\lambda =0$. Following\emph{\ \cite%
{[KP]}} one can observe that this solution $f$ can be found as 
\begin{equation*}
f=c_{1}u_{1}+c_{2}u_{2},
\end{equation*}%
where $c_{1}$ and $c_{2}$ are arbitrary constants, $u_{1}$ and $u_{2}$ are
defined by (\ref{SPPSsolutions})\ with $f\equiv \lambda =1$ and by using $-q$
in place of $r$ in the definition (\ref{K1}) and (\ref{K2}) of $\widetilde{X}%
^{(k)}$ and $X^{(k)}$. In the regular case the choice $c_{1}=1$ and $c_{2}=i$
provides a nonvanishing solution $f$.

The procedure for construction of solutions given by Theorem \ref%
{main_theorem} is still valid when a solution of the equation $(pf^{\ast
\prime })^{\prime }+qf^{\ast }=\lambda ^{\ast }rf^{\ast }$ is available for
some $\lambda =\lambda ^{\ast }$\emph{. }In this case\emph{\ }the solutions (%
\ref{SPPSsolutions}) take the form 
\begin{equation}
u_{1}=f^{\ast }\sum_{k=0}^{\infty }\left( \lambda -\lambda ^{\ast }\right)
^{k}\;\widetilde{X}^{(2k)}\text{ \ and }u_{2}=f^{\ast }\sum_{k=0}^{\infty
}\left( \lambda -\lambda ^{\ast }\right) ^{k}\;X^{(2k+1)},
\label{solutions shift}
\end{equation}%
where $\widetilde{X}^{(k)}$ and $X^{(k)}$ are given by (\ref{K1}), (\ref{K2}%
) with $f=f^{\ast }$. Indeed, equation (\ref{SLlambda}) can be written as 
\begin{equation*}
(pu^{\prime })^{\prime }+\left( q-\lambda ^{\ast }r\right) u=\left( \lambda
-\lambda ^{\ast }\right) ru,
\end{equation*}%
then the same arguments that were used to prove Proposition \ref{Xestimates}
and Theorem \ref{main_theorem} can be applied. The procedure for
constructing solutions of equation (\ref{SLlambda}) based on a particular
solution $f^{\ast }$ corresponding to $\lambda =\lambda ^{\ast }$ is known
as the\emph{\ }spectral shift \ and is of great practical importance
especially in numerical applications\emph{\ \cite{[KP]}, \cite{KKB2012}, 
\cite{KrTV}.}
\end{remark}

\section{Numerical implementation and examples}

The SPPS method based on the representations established in Theorem \ref%
{main_theorem} is especially convenient for solving spectral problems. It
has been used previously in a number of works and proved to be efficient in
various applications both with continuous \cite{CKKO}, \cite{KKB2012}, \cite%
{BarBlKr}, \cite{BarKrR}, \cite{CastKrOR}, \cite{CKT}, \cite{KhKrR}, \cite%
{KrTV}, \cite{KrV} and some special cases of singular coefficients \cite%
{CastKrT}. In this section we present several numerical examples
illustrating the application of the method to the spectral problems with
discontinuous coefficients. In practical terms the SPPS method can be
formulated in the following steps

\begin{itemize}
\item obtain an analytic expression for the characteristic function of the
problem which will be denoted by $\Phi (\lambda )$ in terms of the SPPS
representations (\ref{SPPSsolutions});

\item calculate the first $2N+1$ formal powers $\widetilde{X}^{(n)}$ and $%
X^{(n)}$ necessary to approximate the characteristic function $\Phi (\lambda
)$ by a partial sum $\Phi _{N}(\lambda )$;

\item find roots of the equation $\Phi _{N}(\lambda )=0$.
\end{itemize}

To perform the second step it is necessary to find a nonvanishing particular
solution $f$ of equation (\ref{SLhom}) wich can be obtained by the SPPS
method (see Remark \ref{Shift lambda}). The SPPS representation given by
Theorem \ref{main_theorem} is based on this particular solution $f$ \ which
corresponds to $\lambda =0$ and the power series of $\Phi (\lambda )$ is
centered in $\lambda =0$. To improve the accuracy of the eigenvalues located
farther from the center of the series we perform the spectral shift
described in Remark \ref{Shift lambda}. On every step after calculating an
eigenvalue $\lambda _{n}$ this value is chosen as a new $\lambda ^{\ast }$
and the corresponding particular solution $f^{\ast }$ is computed according
to (\ref{solutions shift}). In the case when the boundary conditions are
spectral parameter dependent, the characteristic function can again be
written in a form containing power series in terms of $(\lambda -\lambda
^{\ast })$, we illustrate it in some examples below.

The range of applicability of the SPPS method includes complex coefficients
and complex eigenvalues. Therefore it is important to note that step 3 can
be complemented with a preliminary counting of zeros in a given domain of
the complex plane of the variable $\lambda $. This can be done using a
classical result from complex analysis - the principle of the argument. More
on applications of the argument principle (as well as of Rouche's theorem)
can be found in \cite{CastKrT}, \cite{KrTV}.

All calculations were performed with the aid of Matlab 2009 in the double
precision machine arithmetic. The formal powers $\widetilde{X}^{(n)}$ and $%
X^{(n)}$ were computed using the Newton-Cottes 6 point integration formula
of 7th order, modified to implement indefinite integration. The integration
was performed separately on each subinterval where the coefficients of the
considered equation were continuous, with subsequent joining together of
separate integrals into a continuous function. To find zeros of the
polynomial $\Phi _{N}(\lambda )$ the routine \texttt{roots }of \textsc{Matlab%
} was used.

In all the considered examples we use the spectral shift technique (see
Remark \ref{Shift lambda}) for calculating every subsequent eigenvalue $%
\lambda _{n}$ with $\lambda _{n}^{\ast }=\lambda _{n-1}+\Delta \lambda $,
where $\Delta \lambda $ is a displacement and $\lambda _{0}^{\ast }=0$.

In the presented numerical results we specify two parameters: $N$ is the
degree of the polynomial $\Phi _{N}$, i.e. the number of calculated formal
powers is $2N+1$, and $M$ is the number of points chosen on the considered
segment for the calculation of integrals. To obtain the exact eigenvalues we
use the routine \texttt{FindRoot} of \textsc{Mathematica }applied to the
exact characteristic function $\Phi (\lambda )$\textsc{.}

\begin{example}
\label{example1}\bigskip Our first example is taken from \cite{Tharwat}
where the numerical results are obtained by means of the sinc method.
Consider the equation%
\begin{equation}
-u^{\prime \prime }+qu=\lambda u,\quad x\in \lbrack -1,1],  \label{equation1}
\end{equation}%
where%
\begin{equation*}
q(x)=\left\{ 
\begin{array}{ll}
-1, & x\in \lbrack -1,0], \\ 
-2, & x\in (0,1]%
\end{array}%
\right. 
\end{equation*}%
with the boundary conditions 
\begin{eqnarray}
\lambda u(-1)+u^{\prime }(-1) &=&0,  \label{boundary1} \\
\lambda u(1)-u^{\prime }(1) &=&0,  \notag
\end{eqnarray}%
and the transmission conditions%
\begin{equation*}
u(0_{-})=u(0_{+}),\quad u^{\prime }(0_{-})=u^{\prime }(0_{+}).
\end{equation*}%
The subscripts "$+$" and "$-$" denote the limiting values of $u(x)$ as $x$
approaches $0$ from the right and left, respectively. A general solution of
equation (\ref{equation1}) ($\lambda \neq -2$) which satisfies the
transmission conditions\ is%
\begin{equation*}
u(x)=\left\{ 
\begin{array}{cc}
A\cos \sqrt{1+\lambda }x+B\frac{\sqrt{2+\lambda }}{\sqrt{1+\lambda }}\sin 
\sqrt{1+\lambda }x, & x\in \lbrack -1,0], \\ 
A\cos \sqrt{2+\lambda }x+B\sin \sqrt{2+\lambda }x, & x\in (0,1].%
\end{array}%
\right. 
\end{equation*}%
Substituting it to the boundary conditions (\ref{boundary1}) it is easy to
arrive at the following characteristic equation%
\begin{gather*}
\Phi (\lambda )=\frac{1}{\sqrt{2+\lambda }}\left\{ \cos \sqrt{1+\lambda }%
\left[ \left( \lambda ^{2}-\lambda -2\right) \sin \sqrt{2+\lambda }-2\lambda 
\sqrt{2+\lambda }\cos \sqrt{2+\lambda }\right] \right. + \\
\left. \frac{1}{\sqrt{1+\lambda }}\sin \sqrt{1+\lambda }\left[ \lambda
\left( 2\lambda +3\right) \sin \sqrt{2+\lambda }+\sqrt{2+\lambda }\left(
\lambda ^{2}-\lambda -1\right) \cos \sqrt{2+\lambda }\right] \right\} =0
\end{gather*}%
In terms of the SPPS solutions (\ref{SPPSsolutions}) the eigenfunctions of
the spectral problem take the form%
\begin{equation*}
u(x,\lambda )=u_{1}(x,\lambda )+f(-1)(\lambda f(-1)+f^{\prime
}(-1))u_{2}(x,\lambda ),
\end{equation*}%
where $f(x)=\left\{ 
\begin{array}{cc}
\cos x & x\in \lbrack -1,0] \\ 
\cos \sqrt{2}x & x\in (0,1]%
\end{array}%
\right. $ and $\lambda $ satisfies the SPPS characteristic equation 
\begin{equation*}
\Phi (\lambda )=(\lambda f(-1)+f^{\prime }(-1))\left( \lambda
u_{2}(1,\lambda )-u_{2}^{\prime }(1,\lambda )\right) +\frac{1}{f(-1)}\left(
\lambda u_{1}(1,\lambda )-u_{1}^{\prime }(1,\lambda )\right) =0.
\end{equation*}%
We denote for short $f_{\pm 1}:=f(\pm 1)$, $\widetilde{X}^{(n)}(1):=%
\widetilde{X}_{1}^{(n)}$, $X^{(n)}(1):=X_{1}^{(n)}$\ and taking into account
that $p(x)=-1$ write down $\Phi (\lambda )$ in explicit form 
\begin{gather*}
\Phi (\lambda )=\sum_{k=0}^{\infty }\lambda ^{k}C_{k}\text{, where} \\
C_{k}=f_{-1}\left( f_{1}X_{1}^{(2k-3)}+\frac{X_{1}^{(2k-2)}}{f_{1}}%
-f_{1}^{\prime }X_{1}^{(2k-1)}\right) +f_{-1}^{\prime }\left(
f_{1}X_{1}^{(2k-1)}+\frac{X_{1}^{(2k)}}{f_{1}}-f_{1}^{\prime
}X_{1}^{(2k+1)}\right) + \\
\frac{1}{f_{-1}}\left( f_{1}\widetilde{X}_{1}^{(2k-2)}+\frac{\widetilde{X}%
_{1}^{(2k-1)}}{f_{1}}-f_{1}^{\prime }\widetilde{X}_{1}^{(2k)}\right) ,\text{
note that }X_{1}^{(\alpha )}\text{, }\widetilde{X}_{1}^{(\alpha )}\text{
equal cero for }\alpha <0.
\end{gather*}

When the shift by $\lambda ^{\ast }$ is performed the characteristic
function can be written as power series in terms of $(\lambda -\lambda
^{\ast })$, we denote it by $\Phi ^{\ast }(\lambda -\lambda ^{\ast })$ 
\begin{gather*}
\Phi ^{\ast }(\lambda -\lambda ^{\ast })=\Phi (\lambda -\lambda ^{\ast
})+\lambda ^{\ast }\sum_{k=0}^{\infty }(\lambda -\lambda ^{\ast })^{k}B_{k}%
\text{, where } \\
B_{k}=f_{-1}^{\ast }\left( 2f_{1}^{\ast }X_{1}^{(2k-1)}-f_{1}^{\ast \prime
}X_{1}^{(2k+1)}+\frac{X_{1}^{(2k)}}{f_{1}^{\ast }}\right) +\left( \lambda
^{\ast }f_{-1}^{\ast }+f_{-1}^{\ast \prime }\right) f_{1}^{\ast
}X_{1}^{(2k+1)}+\frac{f_{1}^{\ast }}{f_{-1}^{\ast }}\widetilde{X}_{1}^{(2k)}
\end{gather*}
and $X_{1}^{(-1)}=0$. The formal powers $\widetilde{X}^{(k)}$and $X^{(k)}$
are calculated again by (\ref{K1}) and (\ref{K2}) but now with $f=f^{\ast }$%
.\medskip

In Table 1 we present the exact eigenvalues obtained by \texttt{FindRoot}%
\textsc{,} an absolute error of the numerically computed eigenvalues by
means of the SPPS method and the result from \cite{Tharwat} where the sinc
method is used to calculate the square roots of the eigenvalues.\medskip

\begin{tabular}{llll}
\multicolumn{4}{l}{\textbf{Table 1.} The exact eigenvalues from example \ref%
{example1}} \\ 
\multicolumn{4}{l}{and their absolute error obtained with $N=60$, $M=50000$}
\\ \hline\hline
$n$ & Exact $\lambda _{n}$ & SPPS abs. error & sinc abs. error \cite{Tharwat}
\\ \hline
$0$ & $-0.8838501773806790$ & $7.4027\times 10^{-15}$ &  \\ \hline
$1$ & $0.33593977069858758$ & $5.7176\times 10^{-15}$ & $7.\,337\,3\times
10^{-13}$ \\ \hline
$2$ & $3.18616750501251774$ & $1.1941\times 10^{-13}$ & $4.306\,1\times
10^{-12}$ \\ \hline
$3$ & $10.4888366560518901$ & $1.5446\times 10^{-13}$ & $7.819\,4\times
10^{-12}$ \\ \hline
$4$ & $22.7582649977549487$ & $5.6014\times 10^{-13}$ & $6.103\,9\times
10^{-11}$ \\ \hline
$5$ & $40.0145357092335736$ & $4.8143\times 10^{-13}$ & $1.296\,3\times
10^{-10}$ \\ \hline
$6$ & $62.2048900366600122$ & $1.4728\times 10^{-12}$ &  \\ \hline
$7$ & $89.3430782636483577$ & $1.3721\times 10^{-12}$ &  \\ \hline
$8$ & $121.412971555997595$ & $1.6146\times 10^{-12}$ &  \\ \hline
$9$ & $158.423147639717927$ & $2.1134\times 10^{-12}$ &  \\ \hline
$10$ & $200.365776764230126$ & $1.8909\times 10^{-12}$ &  \\ \hline
$20$ & $891.233220344783089$ & $6.1716\times 10^{-12}$ &  \\ \hline
$30$ & $2075.58493709348608$ & $8.6609\times 10^{-12}$ &  \\ \hline
$50$ & $5924.73025618735463$ & $1.5332\times 10^{-11}$ &  \\ \hline
$75$ & $13511.9885376287504$ & $1.6978\times 10^{-11}$ &  \\ \hline
$100$ & $24183.4982363150305$ & $3.3070\times 10^{-10}$ &  \\ \hline
\end{tabular}%
\medskip

Notice that in \cite{Tharwat} the negative eigenvalue was not detected .
\end{example}

\begin{example}
\label{example2}As a second example we consider the eigenvalue problem that
arises \cite{KhS} in the study of heat conduction in layered composites $\ $%
\begin{equation*}
-(pu^{\prime })^{\prime }=\lambda ru,\quad x\in \lbrack a_{1},a_{4}]
\end{equation*}%
where 
\begin{equation*}
p(x)=\left\{ 
\begin{tabular}{ll}
$p_{1},$ & $x\in \left[ a_{1},a_{2}\right) $ \\ 
$p_{2},$ & $x\in \left[ a_{2},a_{3}\right) $ \\ 
$p_{3}$ & $x\in \left[ a_{3},a_{4}\right] $%
\end{tabular}%
\right. ,\quad r(x)=\left\{ 
\begin{tabular}{ll}
$r_{1},$ & $x\in \left[ a_{1},a_{2}\right) $ \\ 
$r_{2},$ & $x\in \left[ a_{2},a_{3}\right) $ \\ 
$r_{3}$ & $x\in \left[ a_{3},a_{4}\right] $%
\end{tabular}%
\right.
\end{equation*}%
with $p_{i},$ $r_{i}$ being nonzero complex-valued constants, the boundary
and the transmission conditions are%
\begin{equation*}
u(a_{1})=u(a_{4})=0,
\end{equation*}%
\begin{eqnarray*}
u(a_{2_{-}}) &=&u(a_{2_{+}}),\quad p_{1}u^{\prime
}(a_{2_{-}})=p_{2}u^{\prime }(a_{2_{+}}), \\
u(a_{3_{-}}) &=&u(a_{3_{+}}),\quad p_{2}u^{\prime
}(a_{3_{-}})=p_{3}u^{\prime }(a_{3_{+}}).
\end{eqnarray*}%
It can be verified that the eigenfunction of this Sturm-Liouville problem
for $\lambda \neq 0$ can be taken in the form%
\begin{equation*}
u(x)=\left\{ 
\begin{tabular}{ll}
$\sin \sqrt{\frac{\lambda r_{1}}{p_{1}}}\left( x-a_{1}\right) ,$ & $x\in %
\left[ a_{1},a_{2}\right] $ \\ 
$v(x),$ & $x\in \left[ a_{2},a_{3}\right] $ \\ 
$\frac{v(a_{3})\sin \sqrt{\frac{\lambda r_{3}}{p_{3}}}\left( x-a4\right) }{%
\sin \sqrt{\frac{\lambda r_{3}}{p_{3}}}\left( a_{3}-a_{4}\right) },$ & $x\in %
\left[ a_{3},a_{4}\right] $%
\end{tabular}%
\right. ,
\end{equation*}%
where $v(x)=\sin \sqrt{\frac{\lambda r_{1}}{p_{1}}}\left( a_{2}-a_{1}\right)
\cos \sqrt{\frac{\lambda r_{2}}{p_{2}}}\left( x-a_{2}\right) +\sqrt{\frac{%
p_{1}r_{1}}{p_{2}r_{2}}}\cos \sqrt{\frac{\lambda r_{1}}{p_{1}}}\left(
a_{2}-a_{1}\right) \sin \sqrt{\frac{\lambda r_{2}}{p_{2}}}\left(
x-a_{2}\right) $, for values of $\lambda $ satisfying the following
characteristic equation which is a result of the transmission condition $%
p_{2}u^{\prime }(a_{3_{-}})=p_{3}u^{\prime }(a_{3_{+}})$ 
\begin{equation*}
\Phi (\lambda )=p_{2}v^{\prime }(a_{3})-p_{3}\frac{v(a_{3})\sqrt{\frac{%
\lambda r_{3}}{p_{3}}}\cos \sqrt{\frac{\lambda r_{3}}{p_{3}}}(a_{4}-a_{3})}{%
\sin \sqrt{\frac{\lambda r_{3}}{p_{3}}}(a_{3}-a_{4})}=0
\end{equation*}%
In terms of the SPPS solutions (\ref{SPPSsolutions}) the eigenfunctions of
this spectral problem are 
\begin{equation*}
u(x,\lambda )=u_{2}(x,\lambda )=\sum_{k=0}^{\infty }\lambda
^{k}\;X^{(2k+1)}(x),
\end{equation*}%
where $\lambda $ satisfies the characteristic equation%
\begin{equation*}
\Phi (\lambda )=u_{2}(a_{4},\lambda )=0.
\end{equation*}%
\ \ In this example, we calculate the eigenvalues for the cases of both real
and complex coefficients. The results for the case of real coefficients is
presented in Table 2 which contain the exact eigenvalues obtained by \texttt{%
FindRoot}\textsc{\ }and an absolute error of the numerically computed
eigenvalues by means of the SPPS method. For the calculations the following
values of parameters were used $a_{1}=-4;\ a_{2}=-2;\ a_{3}=0;\ a_{4}=2;\
p_{1}=11;\ p_{2}=0.5;\ p_{3}=22;\ r_{1}=3;\ r_{2}=7;\ r_{3}=1\medskip $

\begin{tabular}{lll}
\multicolumn{3}{l}{\textbf{Table 2.} The exact eigenvalues from example \ref%
{example2}} \\ 
\multicolumn{3}{l}{and their absolute error obtained with $N=90$, $M=149998$}
\\ \hline\hline
$n$ & Exact $\lambda _{n}$ & SPPS abs. error \\ \hline
$0$ & $0.1537166881459068$ & $8.2712\times 10^{-15}$ \\ \hline
$1$ & $0.6040510821002027$ & $3.2350\times 10^{-14}$ \\ \hline
$2$ & $1.3001446415922297$ & $5.3654\times 10^{-14}$ \\ \hline
$3$ & $2.1131346987303714$ & $8.5513\times 10^{-14}$ \\ \hline
$4$ & $3.0657222557870770$ & $4.3681\times 10^{-14}$ \\ \hline
$5$ & $4.3891432656424323$ & $1.0272\times 10^{-13}$ \\ \hline
$6$ & $6.0755689904595746$ & $9.1154\times 10^{-14}$ \\ \hline
$7$ & $8.0532227313898138$ & $3.1959\times 10^{-13}$ \\ \hline
$8$ & $10.263818816222202$ & $1.6454\times 10^{-13}$ \\ \hline
$9$ & $12.662474776451201$ & $2.3047\times 10^{-13}$ \\ \hline
$10$ & $15.226226392993377$ & $4.5646\times 10^{-13}$ \\ \hline
$20$ & $54.709035491349110$ & $3.5327\times 10^{-10}$ \\ \hline
$30$ & $117.46422540410899$ & $7.7278\times 10^{-10}$ \\ \hline
$50$ & $320.94806622153616$ & $2.3762\times 10^{-7}$ \\ \hline
$75$ & $706.76377022340156$ & $2.8990\times 10^{-7}$ \\ \hline
$100$ & $1249.2350537091337$ & $1.7486\times 10^{-4}$ \\ \hline
\end{tabular}%
$\medskip $

For the second case of complex coefficients we take the same interval as in
the previous case and the following complex values of parameters $\
p_{1}=11+1i;\ p_{2}=0.5+2i;\ p_{3}=22+1i;\ r_{1}=3+2i;\ r_{2}=7+1i;\
r_{3}=1-2i$. The numerical result presented in Table 3 is obtained using the
spectral shift with $\lambda _{n}^{\ast }=\lambda _{n-1}+0.5$ if $\func{Im}%
(\lambda _{n-1})>0$. If $\func{Im}(\lambda _{n-1})<0$, we do not make a
shift, i.e., $\lambda _{n}^{\ast }=\lambda _{n-2}+0.5$. The computation time
for each obtained eigenvalue was around 20 sec.

\bigskip 
\begin{tabular}{lll}
\multicolumn{3}{l}{\textbf{Table 3.} The exact eigenvalues from example \ref%
{example2}} \\ 
\multicolumn{3}{l}{and their absolute error obtained with $N=90$, $M=120001$}
\\ \hline\hline
$n$ & Exact $\lambda _{n}$ & SPPS abs. error \\ \hline
$0$ & $0.469982057297078+0.337010475999479i$ & $9.7688\times 10^{-14}$ \\ 
\hline
$1$ & $1.453180135224583+0.455435050626238i$ & $1.6122\times 10^{-13}$ \\ 
\hline
$2$ & $1.931066258100073+1.957548941283227i$ & $1.4746\times 10^{-13}$ \\ 
\hline
$3$ & $2.769261458131468+4.456326784352162i$ & $4.6264\times 10^{-14}$ \\ 
\hline
$4$ & $3.315488435122103+8.156636278096363i$ & $2.0441\times 10^{-13}$ \\ 
\hline
$5$ & $4.745130735885916+11.83858259923195i$ & $2.2416\times 10^{-13}$ \\ 
\hline
$6$ & $14.63113794579346-3.537417383243752i$ & $1.4060\times 10^{-12}$ \\ 
\hline
$7$ & $7.897123133993671+17.11862172171793i$ & $1.8218\times 10^{-13}$ \\ 
\hline
$8$ & $11.68836415028633+23.94806911602541i$ & $4.1506\times 10^{-13}$ \\ 
\hline
$9$ & $16.17740169627085+31.68384253976561i$ & $1.0887\times 10^{-12}$ \\ 
\hline
$10$ & $41.18479452885688-14.80403116625259i$ & $1.6835\times 10^{-11}$ \\ 
\hline
$14$ & $28.58271526370320+69.03918582895608i$ & $9.0925\times 10^{-12}$ \\ 
\hline
$15$ & $81.02445701222293-33.28304821488122i$ & $2.3535\times 10^{-9}$ \\ 
\hline
$16$ & $35.59337686888703+81.52159363573743i$ & $2.0417\times 10^{-12}$ \\ 
\hline
$21$ & $61.04946135834132+140.2136938161330i$ & $5.1677\times 10^{-8}$ \\ 
\hline
$22$ & $64.21517057997782+157.3867636446106i$ & $1.0209\times 10^{-6}$ \\ 
\hline
$25$ & $200.5617997525851-91.87829196068057i$ & $3.1432\times 10^{-3}$ \\ 
\hline
$26$ & $94.66606874013544+215.4400279751237i$ & $1.3002\times 10^{-6}$ \\ 
\hline
\end{tabular}

Let us note that the SPPS approach allows one to visualize the
characteristic function of the problem. In Fig. 1 we plot the SPPS
approximation of the function $-\ln \left\vert \Phi _{90}(\lambda
)\right\vert $ on the disk $\left\vert \lambda \right\vert \leqslant 45$ in
the complex plane of the variable $\lambda $. The peaks on the graph
correspond to the first eleven approximate eigenvalues of the problem. 
\FRAME{ftbpFU}{6.0096in}{2.77in}{0pt}{\Qcb{The graph of $-\log \left\vert
\Phi _{90}(\protect\lambda )\right\vert $ in the disk $\left\vert \protect%
\lambda \right\vert \leqslant 45$. The peaks represent the first eleven
approximate eigenvalues presented in Table 3.}}{}{dicont_complex.eps}{%
\special{language "Scientific Word";type "GRAPHIC";display
"USEDEF";valid_file "F";width 6.0096in;height 2.77in;depth
0pt;original-width 9.3452in;original-height 6.5337in;cropleft "0";croptop
"1";cropright "1";cropbottom "0";filename
'graficas/Dicont_complex.eps';file-properties "XNPEU";}}
\end{example}

\begin{example}
\label{example3}Consider the eigenvalue problem%
\begin{eqnarray}
-u^{\prime \prime }+qu &=&\lambda u,\quad x\in \lbrack -1,1],
\label{eq_example3} \\
u(-1)+\lambda u^{\prime }(-1) &=&0,\quad u(1)+\lambda u^{\prime }(1)=0, 
\notag
\end{eqnarray}%
with 
\begin{equation*}
q(x)=\left\{ 
\begin{array}{ll}
-2, & x\in \lbrack -1,0], \\ 
\ \ x, & x\in (0,1].%
\end{array}%
\right.
\end{equation*}%
This example is from \cite{Tharwat} where unfortunately the results are
given with a misprint. Due to this we are not able to compare our numerical
results to those presented in \cite{Tharwat}.

It is easy to see that the general solution of the equation (\ref%
{eq_example3}) is%
\begin{equation*}
u(x,\lambda )=\left\{ 
\begin{array}{cc}
\begin{array}{c}
\left( C_{1}Ai(-\lambda )+C_{2}Bi(-\lambda )\right) \cos \sqrt{2+\lambda }x+
\\ 
\frac{1}{\sqrt{2+\lambda }}\left( C_{1}Ai^{\prime }(-\lambda
)+C_{2}Bi^{\prime }(-\lambda )\right) \sin \sqrt{2+\lambda }x,%
\end{array}
& x\in \lbrack -1,0], \\ 
C_{1}Ai(x-\lambda )+C_{2}Bi(x-\lambda ), & x\in \lbrack 0,1],%
\end{array}%
\right.
\end{equation*}%
where $Ai$ and $Bi$ are the Airy functions. The exact characteristic
equation for this problem takes the form%
\begin{gather*}
\Phi (\lambda )=\left( \left( Ai(-\lambda )+\lambda Ai^{\prime }(-\lambda
)\right) \cos \sqrt{2+\lambda }+\left( \lambda \sqrt{2+\lambda }Ai(-\lambda
)-\frac{Ai^{\prime }(-\lambda )}{\sqrt{2+\lambda }}\right) \sin \sqrt{%
2+\lambda }\right) \\
\left( Bi(1-\lambda )+\lambda Bi^{\prime }(1-\lambda )\right) - \\
\left( \left( Bi(-\lambda )+\lambda Bi^{\prime }(-\lambda )\right) \cos 
\sqrt{2+\lambda }+\left( \lambda \sqrt{2+\lambda }Bi(-\lambda )-\frac{%
Bi^{\prime }(-\lambda )}{\sqrt{2+\lambda }}\right) \sin \sqrt{2+\lambda }%
\right) \\
\left( Ai(1-\lambda )+\lambda Ai^{\prime }(1-\lambda )\right) =0
\end{gather*}%
The characteristic equation in terms of the SPPS method is%
\begin{equation*}
\Phi (\lambda )=\left( f(-1)+\lambda f^{\prime }(-1)\right) \left(
u_{2}(1)+\lambda u_{2}^{\prime }(1)\right) -\frac{\lambda }{p(-1)f(-1)}%
\left( u_{1}(1)+\lambda u_{1}^{\prime }(1)\right) =0,
\end{equation*}%
where $f$ is a nonvanishing particular solution of equation (\ref%
{eq_example3}) for $\lambda =0$ 
\begin{equation*}
f(x)=\left\{ 
\begin{array}{cc}
\cos \sqrt{2}x, & x\in \lbrack -1,0], \\ 
\frac{Bi^{\prime }(0)}{W(0)}Ai(x)-\frac{Ai^{\prime }(0)}{W(0)}Bi(x), & x\in
\lbrack 0,1],%
\end{array}%
\right.
\end{equation*}%
$W(0)$ denotes the Wronskian of the functions $Ai$ and $Bi$ evaluated in
zero.

Thus, $\Phi (\lambda )$ can be written as power series in terms of $\lambda $
\begin{gather*}
\Phi (\lambda )=\sum_{k=0}^{\infty }\lambda ^{k}C_{k}\text{, where} \\
C_{k}=f_{-1}\left( f_{1}\;X_{1}^{(2k+1)}+f_{1}^{\prime }\;X_{1}^{(2k-1)}+%
\frac{X_{1}^{(2k-2)}}{pf_{1}}\;\right) + \\
f_{-1}^{\prime }\left( f_{1}\;X_{1}^{(2k-1)}+f_{1}^{\prime }\;X_{1}^{(2k-3)}+%
\frac{X_{1}^{(2k-4)}}{pf_{1}}\right) +\frac{1}{f_{-1}}\left( f_{1}\widetilde{%
X}_{1}^{(2k-2)}+f_{1}^{\prime }\;\widetilde{X}_{1}^{(2k-4)}+\frac{\widetilde{%
X}_{1}^{(2k-5)}}{pf_{1}}\;\right)
\end{gather*}%
To perform the shift by $\lambda ^{\ast }$ we write the characteristic
function in terms of $(\lambda -\lambda ^{\ast })$ and denote it as $\Phi
^{\ast }(\lambda -\lambda ^{\ast })$%
\begin{equation*}
\Phi ^{\ast }(\lambda -\lambda ^{\ast })=\Phi (\lambda -\lambda ^{\ast
})+\lambda ^{\ast }\sum_{k=0}^{\infty }(\lambda -\lambda ^{\ast })^{k}B_{k}%
\text{,}
\end{equation*}%
where $B_{k}=\left( f_{-1}^{\ast }+\lambda ^{\ast }f_{-1}^{\ast \prime
}\right) \left( f_{1}^{\ast \prime }X_{1}^{(2k+1)}-\frac{X_{1}^{(2k)}}{%
f_{1}^{\ast }}\right) +f_{-1}^{\ast \prime }\left( f_{1}^{\ast
}X_{1}^{(2k+1)}+2f_{1}^{\ast \prime }X_{1}^{(2k-1)}-\frac{2X_{1}^{(2k-2)}}{%
f_{1}^{\ast }}\right) +$

$\frac{1}{f_{-1}^{\ast }}\left( f_{1}^{\ast }\widetilde{X}%
_{1}^{(2k)}+2f_{1}^{\ast \prime }\widetilde{X}_{1}^{(2k-2)}-\frac{2%
\widetilde{X}_{1}^{(2k-3)}}{f_{1}^{\ast }}+\lambda ^{\ast }\left(
f_{1}^{\ast \prime }\widetilde{X}_{1}^{(2k)}-\frac{\widetilde{X}_{1}^{(2k-1)}%
}{f_{1}^{\ast }}\right) \right) $, $k=0,1,2,...\medskip $and again $%
\widetilde{X}_{1}^{(\alpha )}$ and $X_{1}^{(\alpha )}$ equal zero for $%
\alpha <0$ and the formal powers $\widetilde{X}^{(k)}$and $X^{(k)}$ are
calculated by (\ref{K1}) and (\ref{K2}) where $f=f^{\ast }$.\medskip

\bigskip 
\begin{tabular}{lll}
\multicolumn{3}{l}{\textbf{Table 4.} The exact eigenvalues from example \ref%
{example3}} \\ 
\multicolumn{3}{l}{and their absolute error obtained with $N=95$, $M=55000$}
\\ \hline\hline
$n$ & Exact $\lambda _{n}$ & SPPS abs. error \\ \hline
$0$ & $-1.00143294415521698407$ & $5.5511\times 10^{-15}$ \\ \hline
$1$ & $2.4057972392439196797808$ & $3.3793\times 10^{-13}$ \\ \hline
$2$ & $9.1124600099908036194275$ & $8.6482\times 10^{-12}$ \\ \hline
$3$ & $21.519631798576724032730$ & $1.5888\times 10^{-11}$ \\ \hline
$4$ & $38.723530941627280094388$ & $8.7765\times 10^{-11}$ \\ \hline
$5$ & $60.956434348891755464346$ & $1.6287\times 10^{-10}$ \\ \hline
$6$ & $88.074068137541119404825$ & $3.3656\times 10^{-10}$ \\ \hline
$7$ & $120.16339550734279907486$ & $4.9152\times 10^{-10}$ \\ \hline
$8$ & $157.16229901349984050153$ & $8.4092\times 10^{-10}$ \\ \hline
$9$ & $199.11594231603456321312$ & $9.8560\times 10^{-10}$ \\ \hline
$10$ & $245.98922143170974248778$ & $1.9757\times 10^{-9}$ \\ \hline
$20$ & $986.21021437494101997432$ & $3.6385\times 10^{-8}$ \\ \hline
$30$ & $2219.9108896264072174380$ & $2.6014\times 10^{-7}$ \\ \hline
$50$ & $6167.7527144043198013992$ & $4.2623\times 10^{-4}$ \\ \hline
\end{tabular}
\end{example}

\begin{example}
\label{example4}This example is taken from the list of test problems
presented in \cite{Pryce}.\ Consider the eigenvalue problem%
\begin{eqnarray*}
-u^{\prime \prime }+u &=&\lambda ru,\quad x\in \lbrack 0,1]\medskip \\
u(0) &=&u(1)=0\medskip ,
\end{eqnarray*}%
where

\begin{equation*}
r(x)=\left\{ 
\begin{tabular}{ll}
$0,$ & $x\in \left[ 0,\frac{1}{2}\right] $ \\ 
$1,$ & $x\in \left( \frac{1}{2},1\right] $%
\end{tabular}%
\right. .
\end{equation*}%
The exact characteristic equation for this problem is 
\begin{equation*}
\Phi (\lambda )=\tan \frac{\sqrt{\lambda -1}}{2}+\sqrt{\lambda -1}\tanh 
\frac{1}{2}=0,
\end{equation*}%
whereas the SPPS characteristic equation has the form%
\begin{equation*}
\Phi (\lambda )=u_{2}(1,\lambda )=0.
\end{equation*}%
The numerical result is presented in Table 5.

\bigskip 
\begin{tabular}{lll}
\multicolumn{3}{l}{\textbf{Table 5.} The exact eigenvalues from example \ref%
{example4}} \\ 
\multicolumn{3}{l}{and their absolute error obtained with $N=40$, $M=300000$}
\\ \hline\hline
$n$ & Exact $\lambda _{n}$ & SPPS abs. error \\ \hline
$0$ & $17.89793137541756$ & $1.5632\times 10^{-13}$ \\ \hline
$1$ & $98.16027543604447$ & $1.9401\times 10^{-11}$ \\ \hline
$2$ & $256.2710801437674$ & $2.7910\times 10^{-11}$ \\ \hline
$3$ & $493.2013196148296$ & $8.4422\times 10^{-11}$ \\ \hline
$4$ & $809.0540168683802$ & $1.6823\times 10^{-10}$ \\ \hline
$5$ & $1203.851208314645$ & $1.5561\times 10^{-10}$ \\ \hline
$6$ & $1677.599761311537$ & $7.0518\times 10^{-10}$ \\ \hline
$7$ & $2230.302360623276$ & $1.0313\times 10^{-9}$ \\ \hline
$8$ & $2861.960227826747$ & $3.1004\times 10^{-10}$ \\ \hline
$9$ & $3572.573982714352$ & $2.6569\times 10^{-9}$ \\ \hline
$10$ & $4362.143966674130$ & $2.3226\times 10^{-9}$ \\ \hline
$11$ & $5230.670380220453$ & $1.4316\times 10^{-6}$ \\ \hline
$12$ & $6178.153347366686$ & $5.9094\times 10^{-4}$ \\ \hline
$13$ & $7204.592948137174$ & $5.2852\times 10^{-4}$ \\ \hline
$14$ & $8309.989236037242$ & $2.4735\times 10^{-3}$ \\ \hline
\end{tabular}

Despite a large number of points used in this example for integrations, the
computational time for each obtained eigenvalue was around 20 sec.
\end{example}

\end{document}